\documentclass{article}
\usepackage[T1]{fontenc}
\usepackage[latin9]{inputenc}
\usepackage{geometry}
\usepackage{color}
\usepackage{mathrsfs}
\usepackage{url}
\usepackage{amsthm}
\usepackage{amsmath}
\usepackage{amssymb}
\usepackage{graphicx}
\usepackage{esint}
\usepackage[unicode=true,
 bookmarks=true,bookmarksnumbered=false,bookmarksopen=false,
 breaklinks=false,pdfborder={0 0 0},backref=page,colorlinks=true]
 {hyperref}
\hypersetup{pdftitle={Spectral Risk Measures},
 pdfauthor={Alois Pichler}}

\makeatletter

\newcommand{\noun}[1]{\textsc{#1}}

\usepackage{enumitem}		
\theoremstyle{plain}
\newtheorem{thm}{\protect\theoremname}
  \theoremstyle{definition}
  \newtheorem{defn}[thm]{\protect\definitionname}
  \theoremstyle{remark}
  \newtheorem{rem}[thm]{\protect\remarkname}
  \theoremstyle{plain}
  \newtheorem{cor}[thm]{\protect\corollaryname}
  \theoremstyle{definition}
  \newtheorem{example}[thm]{\protect\examplename}
  \theoremstyle{plain}
  \newtheorem{lem}[thm]{\protect\lemmaname}

\usepackage{dsfont}   
\usepackage{ifpdf}
\ifpdf		
          \usepackage{pdfsync} 	
          \IfFileExists{lmodern.sty}{\usepackage{lmodern}}{}	
\else		
          \usepackage[active]{srcltx}
\fi

\newcommand{\one}{\mathds 1}  
\newcommand{\CTE}{{\sf CTE}}

\newcommand{\esssup}{\operatornamewithlimits{ess\,sup}}

\@ifundefined{showcaptionsetup}{}{%
 \PassOptionsToPackage{caption=false}{subfig}}
\usepackage{subfig}
\makeatother

  \providecommand{\corollaryname}{Corollary}
  \providecommand{\definitionname}{Definition}
  \providecommand{\examplename}{Example}
  \providecommand{\lemmaname}{Lemma}
  \providecommand{\remarkname}{Remark}
\providecommand{\theoremname}{Theorem}

\begin{document}

\title{Premiums And Reserves, Adjusted By Distortions}

\author{Alois Pichler~%
\thanks{Norwegian University of Science and Technology.\protect \\
Actuary of the Actuarial Association of Austria.\protect \\
Contact: \protect\url{alois.pichler@iot.ntnu.no}%
}}
\maketitle
\begin{abstract}
The net-premium principle is considered to be the most genuine and
fair premium principle in actuarial applications. However, an insurance
company, applying the net-premium principle, goes bankrupt with probability
one in the long run, even if the company covers its entire costs by
collecting the respective fees from its customers. It is therefore
an intrinsic necessity for the insurance industry to apply premium
principles, which guarantee at least further existence of the company
itself; otherwise, the company naturally could not insure its clients
to cover their potential, future claims. Beside this intriguing fact
the underlying loss distribution typically is not known precisely.
Hence alternative premium principles have been developed. A simple
principle, ensuring risk-adjusted credibility premiums, is the distorted
premium principle. This principle is convenient in insurance companies,
as the actuary does not have to change his or her tools to compute
the premiums or reserves. 

This paper addresses the distorted premium principle from various
angles. First, dual characterizations are developed. Next, distorted
premiums are typically computed by under-weighting or ignoring low,
but over-weighting high losses. It is demonstrated here that there
is an alternative, opposite point of view, which consists in leaving
the probability measure unchanged, but increasing the outcomes instead.
It turns out that this new point of view is natural in actuarial practice,
as it can be used for premium calculations, as well as to determine
the reserves of subsequent years in a time consistent way.

\textbf{Classification:} 90C15, 60B05, 62P05

\end{abstract}

\section{Introduction}

Risk adjusted insurance prices by employing \emph{distorted probability
measures }have been considered in this journal by Wang \cite{Wang1998}
and for example in \cite{Balbas2012}. The idea is based on the fact
that outstanding, potential losses should be over-valued, whereas
small claims may be under-weighted in exchange. This procedure provides
a risk-adjusted premium, which always exceeds the net premium (cf.
also the recent papers \cite{Furman2008}). 

In this paper we provide a different perspective in a way, which leaves
the probabilities unchanged (the measure is not changed), but the
claims are adjusted in an appropriate way. Considering just the premium,
then both approaches provide the same result. However, the new perspective
allows computing the reserves as well in a concise and time-consistent
way, and this is the essential novel contribution. 

\medskip{}

Axiomatic characterizations of insurance premiums have been outlined
in \cite{Wang1997}, \cite{Wang2000} and in \cite{Young}. These
axiomatic treatments, initiated in an actuarial context first (early
attempts appeared already in \cite{Denneberg1989}), have been developed
further in financial mathematics, for example in the celebrated seminal
paper \cite{Artzner1999}. The connection between actuarial and financial
mathematics is striking here, as premium principles in an actuarial
context correspond to risk measures in financial mathematics, so that
risk measures constitute a premium principle and vice versa. What
perhaps surprises is that the name---risk measure---is a term that
should be expected in actuarial science rather than in financial mathematics.

The distorted  probability relates directly to a special class of
risk measures, the spectral measures introduced in \cite{Acerbi2002}
and \cite{Acerbi2002a}. An important study of spectral risk measures,
although under the different name\emph{ distortion functional}, was
provided in \cite{Pflug2006}. The concepts of (i) premium principles
by distorting probability measures, (ii) distortion functionals, and
(iii) spectral risk measures are essentially the same---they differ
just in sign conventions, resulting in a concave or a convex description.

Distorted premium principles constitute an elementary and important
class of premium principles, as every premium functional can be described
by premium functionals involving distortions. They are moreover defined
in an explicit way, hence there is an explicit evaluation scheme available,
which is of course important for an applied actuary. 

The most important distorted premium functional, which made its way
to the top, is the \emph{conditional tail expectation}, $\CTE$ (in
a financial context the alternative term Conditional Value-at-Risk
is more accepted). The conditional tail expectation is usually associated
and employed for loss distributions of entire portfolios (for example
by the US and Canadian insurance supervisory authorities, \cite{Bangwon2009},
cf. also \cite{Chi2011}). Here we shall exploit that the $\CTE$
constitutes an elementary pricing principle as well (cf.~\cite{Balbas2012}).
It is the essential advantage of the conditional tail expectation
that different representations are known, which makes this premium
principle eligible in varying situations: by conjugate duality there
is an expression in the form of a supremum, but in applications and
for quick computations a different formulation as an infimum is extremely
convenient: developed in the paper \emph{Optimization of Conditional
Value-at-Risk} \cite{RockafellarUryasev2000} (cf. also \cite{Rockafellar}),
the general formula is given in \emph{Some Remarks on the Value-at-Risk
and the Conditional Value-at-Risk} in \cite{Pflug2000}. The main
results of this article extend both formulations to distorted premium
functionals. Further, both representations can be associated with
different views on distortions, providing different interpretations
in an actuarial context.

\medskip{}

A description of the distorted premium principle as a \emph{supremum}
is a first result of this article. The description builds on dual
representations and on second order stochastic dominance. Stochastic
dominance relations have been considered in the literature, but typically
for the risk measure itself (the primal functions) with the negative
result that coherent risk measures---in general---are not consistent
with second order stochastic dominance, cf.~\cite{Giorgi2005,Krokhmal2007}.
A concise formulation, however, is available by imposing stochastic
dominance constraints on the convex conjugate function (the dual function)
instead of considering stochastic orders on the primal (cf. also \cite{Shapiro2011}),
and this is elaborated here. 

Besides that---and this is of particular importance for applications
and a further result in this paper---a formula for the distorted premium
is elaborated by involving an \emph{infimum}. The infimum description
builds on the Fenchel--Young inequality. This alternative representation
of distorted premium functionals is the converse of the initial description,
as it does not change the measure, but the outcomes instead. \medskip{}

The article is organized as follows. The premium principle is introduced
in the following Section~\ref{sec:DistortedDistribution}. Its description
as a supremum by means of stochastic order relations is contained
in Section~\ref{sec:Dual}. The infimum representation is elaborated
in Section~\ref{sec:Infimum-Representation}. Further implications
for actuarial sciences are outlined and explained in Section~\ref{sec:StochasticOptimization},
this section contains illustrating examples as well.

\section{\label{sec:DistortedDistribution}The Distorted Distribution}

In this paper---as usual in an actuarial context---we shall associate
a $\mathbb{R}-$valued random variable with loss and therefore write
$L$ to denote a random variable. $F_{L}\left(x\right):=P\left(L\le x\right)$
is the \emph{cumulative distribution function} (cdf), and 
\begin{equation}
F_{L}^{-1}\left(u\right):=\inf\left\{ x:F_{L}(x)\ge u\right\} \label{eq:1-1}
\end{equation}
 is the \emph{generalized inverse} or \emph{quantile}. The random
variable $L$ can be given by employing the probability integral transform
(or inverse sampling) as 
\begin{equation}
L=F_{L}^{-1}\left(U\right)\quad\text{a.s.},\label{eq:2-1}
\end{equation}
where $U$ is a uniformly distributed random variable%
\footnote{$U$ is \emph{uniformly distributed} if $P\left(U\le u\right)=u$
for all $u\in\left[0,1\right]$.%
} on the same probability space as $L$ and coupled in a co-monotone
way with $L$ (for example $U:=F_{L}\left(L\right)$, if $F_{L}$
is invertible, cf.~\cite{vdVaart}). 

We shall call a nonnegative, nondecreasing function 
\[
\sigma:\left[0,1\right]\to\mathbb{R}_{0}^{+}
\]
satisfying $\int_{0}^{1}\sigma\left(u\right)\mathrm{d}u=1$ \emph{distortion},
and define the antiderivative $\tau_{\sigma}\left(p\right):=\int_{0}^{p}\sigma(u)\mathrm{d}u.$
By the conditions imposed on $\sigma$ the function $\tau_{\sigma}$
is convex, nonnegative and satisfies $\tau_{\sigma}\left(1\right)=1$.
Moreover it has a generalized inverse, $\tau_{\sigma}^{-1}$, defined
in accordance with \eqref{eq:1-1}. 

The \emph{distorted loss} $L_{\sigma}$ (distorted by the distortion
$\sigma$) then is 
\begin{equation}
L_{\sigma}:=F_{L}^{-1}\left(\tau_{\sigma}^{-1}\left(U\right)\right),\label{eq:2-2}
\end{equation}
where $U$ is chosen as in \eqref{eq:2-1}. $L$ and $L_{\sigma}$
notably have the same outcomes, but their probabilities differ. It
holds that $\tau_{\sigma}\left(u\right)\le u$ by monotonicity of
$\sigma$ ($u\in\left[0,1\right]$), such that 
\[
L_{\sigma}\ge L\text{ and }F_{L_{\sigma}}\left(\cdot\right)\le F_{L}\left(\cdot\right)
\]
(it is said that $L_{\sigma}$ stochastically dominates $L$ in first
order). Applying the simple net premium principle to $L_{\sigma}$
and $L$ reveals that 
\[
\mathbb{E}\, L\le\mathbb{E}\, L_{\sigma}=\int_{0}^{1}F_{L}^{-1}\left(\tau_{\sigma}^{-1}\left(u\right)\right)\mathrm{d}u=\int_{0}^{1}F_{L}^{-1}\left(u\right)\mathrm{d}\tau_{\sigma}\left(u\right)=\int_{0}^{1}F_{L}^{-1}\left(u\right)\sigma\left(u\right)\mathrm{d}u
\]
by monotonicity of the expectation, ensuring thus that $\mathbb{E}\, L_{\sigma}$
is a plausible price for the insurance contract, the price $\mathbb{E}\, L_{\sigma}$
at least exceeds the net-premium. 

The premium $\mathbb{E}\, L_{\sigma}$ is moreover easily accessible
to the actuary, because 
\begin{equation}
F_{L_{\sigma}}\left(y\right)=P\left(L_{\sigma}\le y\right)=P\left(U\le\tau_{\sigma}\left(F_{L}\left(y\right)\right)\right)=\tau_{\sigma}\circ F_{L}\left(y\right)\quad\text{a.e.},\label{eq:3}
\end{equation}
the actuary just has to replace the cdf $F_{L}$ by $F_{L_{\sigma}}=\tau_{\sigma}\circ F_{L}$
in his/ her computations for the premium or reserves, or consider
the density 
\[
f_{L_{\sigma}}\left(y\right)=f_{L}\left(y\right)\cdot\sigma\left(F_{L}\left(y\right)\right)
\]
(if available; cf.~\cite{Valdez2011}). So the premium $\mathbb{E}L_{\sigma}$
is an expectation again---as the net premium principle---just with
probabilities modified (distorted) according \eqref{eq:3}. 

\medskip{}

These considerations give rise for the following definition.
\begin{defn}
Let $\sigma\in L^{q}$ ($q\in\left[1,\infty\right]$) be a distortion
and $L\in L^{p}$ be a random variable for the conjugate exponent
$p$ ($\frac{1}{p}+\frac{1}{q}=1$), then
\begin{equation}
\pi_{\sigma}\left(L\right):=\int_{0}^{1}F_{L}^{-1}\left(u\right)\sigma\left(u\right)\mathrm{d}u\label{eq:Rsigma}
\end{equation}
is called $\sigma$\nobreakdash-distorted premium, or simple distorted
premium for the loss $L$. $\pi_{\sigma}$ is called \emph{distorted
premium functional}.\end{defn}
\begin{rem}
The premium $\pi_{\sigma}\left(L\right)$ is well defined and finite
valued, it satisfies $\mathbb{E}\, L\le\pi_{\sigma}\left(L\right)\le\left\Vert \sigma\right\Vert _{q}\cdot\left\Vert L\right\Vert _{p}$
by Hölder's inequality. 
\end{rem}
\medskip{}

The distorted premium functional $\pi_{\sigma}$ satisfies the following
axioms, which have been proposed and formulated in a different context---for
risk measures in mathematical finance---in \cite{Artzner1997}. The
axioms here have been adapted to account for insurance instead of
financial risk (cf. also \cite{Wang1998a}, and for reinsurance cf.~\cite{Balbas2009a}). 
\begin{defn}
\label{Def-RiskMeasure}A function $\pi:L^{p}\to\mathbb{R}$ is called
\emph{premium functional} (or \emph{premium principle}) if the following
axioms are satisfied:
\begin{enumerate}[label=(M)]
\item \emph{\noun{Monotonicity: \label{enu:Monotonicity} }}$\pi\left(L_{1}\right)\le\pi\left(L_{2}\right)$
whenever $L_{1}\le L_{2}$ almost surely;
\end{enumerate}

\begin{enumerate}[label=(C)]
\item \emph{\noun{Convexity: \label{enu:Convexity} }}$\pi\left(\left(1-\lambda\right)L_{0}+\lambda L_{1}\right)\le\left(1-\lambda\right)\,\pi\left(L_{0}\right)+\lambda\,\pi\left(L_{1}\right)$
for $0\le\lambda\le1$;
\end{enumerate}

\begin{enumerate}[label=(T)]
\item \emph{\noun{Translation~equivariance:}}%
\footnote{In an economic or monetary environment this is often called \emph{\noun{Cash
invariance}} instead.%
}\emph{\noun{ \label{enu:Equivariance} }}$\pi\left(L+c\right)=\pi\left(L\right)+c$
if $c\in\mathbb{R}$;
\end{enumerate}

\begin{enumerate}[label=(H)]
\item \emph{\noun{Positive~homogeneity: \label{enu:Homogeneity} }}$\pi\left(\lambda\, L\right)=\lambda\cdot\pi\left(L\right)$
whenever $\lambda>0$.
\end{enumerate}
\end{defn}
\begin{rem}
In a banking or investment environment the interpretation of a reward
is more natural, in this context the mapping $\rho\left(L\right)=\pi\left(-L\right)$
is often considered---and called \emph{coherent risk measure}---instead
(note, that essentially the monotonicity condition \ref{enu:Monotonicity}
and translation property \ref{enu:Equivariance} reverse for $\rho$).

The term \emph{acceptability functional} was introduced in energy
or decision theory to quantify and classify acceptable strategies.
In this context the concave mapping $\mathcal{A}\left(L\right)=-\pi\left(-L\right)$,
the acceptability functional, is employed instead (here, \ref{enu:Convexity}
modifies to concavity). 
\end{rem}
\medskip{}

The conditional tail expectation is the most important premium principle. 
\begin{defn}[Conditional tail expectation]
The premium principle with distortion 
\begin{equation}
\sigma_{\alpha}\left(\cdot\right):=\frac{1}{1-\alpha}\one_{\left(\alpha,1\right]}\left(\cdot\right)\label{eq:sigma}
\end{equation}
is the \emph{conditional tail expectation} at level $\alpha$ ($0\le\alpha<1$),
\[
\CTE_{\alpha}\left(L\right):=\pi_{\sigma_{\alpha}}\left(L\right)=\frac{1}{1-\alpha}\int_{\alpha}^{1}F_{L}^{-1}\left(p\right)\mathrm{d}p.
\]
The conditional tail expectation at level $\alpha=1$ is 
\[
\CTE_{1}\left(L\right):=\lim_{\alpha\nearrow1}\CTE_{\alpha}\left(L\right)=\esssup\, L.
\]

\end{defn}
Due to the defining equation \eqref{eq:Rsigma} of the distorted premium
the same real number is assigned to all random variables $L$ sharing
the same law, irrespective of the underlying probability space. This
gives rise to the notion of version independence: 
\begin{defn}
A premium principle $\pi$ is \emph{version independent}%
\footnote{sometimes also \emph{law invariant} or \emph{distribution based}.%
}, if $\pi\left(L_{1}\right)=\pi\left(L_{2}\right)$ whenever $L_{1}$
and $L_{2}$ share the same law, that is if $P\left(L_{1}\le y\right)=P\left(L_{2}\le y\right)$
for all $y\in\mathbb{R}$.
\end{defn}
The following representation underlines the central role of the conditional
tail expectation for version independent premium principles. Moreover,
it is the basis and justification for investigating distorted premium
principles in much more detail.
\begin{thm}[Kusuoka's representation]
\label{thm:Kusuoka}Any version independent premium principle $\pi$
satisfying \ref{enu:Monotonicity}, \ref{enu:Convexity}, \ref{enu:Equivariance}
and \ref{enu:Homogeneity} on $L^{\infty}$ of an atom-less probability
space has the representation 
\begin{equation}
\pi\left(L\right)=\sup_{\mu\in\mathscr{M}}\int_{0}^{1}\CTE_{\alpha}\left(L\right)\mu\left(\mathrm{d}\alpha\right),\label{eq:Kusuoka}
\end{equation}
where $\mathscr{M}$ is a set of probability measures on $\left[0,1\right]$.\end{thm}
\begin{proof}
Cf.~\cite{Kusuoka,PflugRomisch2007,Shapiro2011} in connection with
\cite{Schachermayer2006}.
\end{proof}
In the present context of distorted premiums it is essential to observe
that any distorted premium has an immediate representation as in \eqref{eq:Kusuoka},
the measure $\mu_{\sigma}$ corresponding to the density $\sigma$
is 
\begin{equation}
\mu_{\sigma}\left(A\right):=\sigma\left(0\right)\delta_{0}\left(A\right)+\int_{A}1-\alpha\:\mathrm{d}\sigma\left(\alpha\right)\qquad(A\subset[0,1],\text{ measurable})\label{eq:1}
\end{equation}
with cumulative distribution function (which we may denote again by
$\mu_{\sigma}$, because it is a measure on $\left[0,1\right]$) 
\[
\mu_{\sigma}\left(p\right)=\left(1-p\right)\sigma\left(p\right)+\int_{0}^{p}\sigma\left(\alpha\right)\mathrm{d}\alpha\quad(0\le p\le1)\qquad(\text{and }\mu_{\sigma}\left(p\right)=0\text{ if }p<0).
\]
$\mu_{\sigma}$ is a positive measure since $\sigma$ is nondecreasing,
and integration by parts reveals that it is a probability measure.
Kusuoka's representation is immediate by Riemann--Stieltjes integration
by parts for the set $\mathscr{M}=\left\{ \mu_{\sigma}\right\} $,
as 
\begin{align*}
\int_{0}^{1}\CTE_{\alpha}\left(L\right)\mu_{\sigma}\left(\mathrm{d}\alpha\right) & =\sigma\left(0\right)\CTE_{0}\left(L\right)+\int_{0}^{1}\frac{1}{1-\alpha}\int_{\alpha}^{1}F_{L}^{-1}\left(p\right)\mathrm{d}p\,\left(1-\alpha\right)\mathrm{d}\sigma\left(\alpha\right)\\
 & =\int_{0}^{1}F_{L}^{-1}\left(p\right)\sigma\left(p\right)\mathrm{d}p=\pi_{\sigma}\left(L\right).
\end{align*}

Conversely, the premium functional $\int_{0}^{1}\CTE_{\alpha}\left(L\right)\mu\left(\mathrm{d}\alpha\right)$
in Kusuoka's representation \eqref{thm:Kusuoka} often can be expressed
as a distorted premium functional with distortion $\sigma_{\mu}$,
this is accomplished by the function 
\begin{equation}
\sigma_{\mu}(\alpha):=\int_{0}^{\alpha}\frac{1}{1-p}\mu\left(\mathrm{d}p\right).\label{eq:3-1}
\end{equation}
Provided that $\sigma_{\mu}$ is well defined (notice that possibly
$\mu\left(\left\{ 1\right\} \right)>0$ has to be excluded when computing
$\sigma_{\mu}\left(1\right)$) it is positive and a density, as $\int_{0}^{1}\sigma_{\mu}(\alpha)\mathrm{d}\alpha=\int_{0}^{1}\frac{1}{1-p}\int_{p}^{1}\mathrm{d}\alpha\mu\left(\mathrm{d}p\right)=1$.

\paragraph*{Kusuoka representation by means of distorted premium principles. }

By the preceding discussion there is a one-to-one relationship $\sigma\mapsto\mu_{\sigma}$
given by \eqref{eq:1} (with inverse $\mu\mapsto\sigma_{\mu}$ given
by \eqref{eq:3-1}) such that Kusuoka's representation (Theorem~\ref{thm:Kusuoka})
can be formulated with distorted premium functionals equally well,
\begin{equation}
\pi\left(L\right)=\sup_{\sigma\in\mathcal{S}}\,\pi_{\sigma}\left(L\right).\label{eq:12-1}
\end{equation}
$\mathcal{S}$ is a set of distortions. $\mathcal{S}$ can be restricted
to consist of continuous and strictly increasing (thus invertible)
density functions. A rigorous discussion is rather straight forward,
although beyond the scope of this article. Here, it is just important
to observe that any premium principle is built of distorted premium
functionals by \eqref{eq:12-1}.

\section{\label{sec:Dual}Supremum-Representation of Distorted Premium Functionals}

The supremum representation of distorted premium functionals is derived
from the convex conjugate relation for convex functionals. To formulate
the result in a concise way we employ the notion of (second order)
stochastic dominance. 
\begin{defn}[Convex ordering]
Let $\tau,\,\sigma:\left[0,1\right]\rightarrow\mathbb{R}$ be integrable
functions.
\begin{enumerate}[label=(\roman{enumi})]
\item $\sigma$ majorizes $\tau$ (denoted $\sigma\succcurlyeq\tau$ or
$\tau\preccurlyeq\sigma$ ) iff 
\[
\int_{\alpha}^{1}\tau\left(p\right)\mathrm{d}p\le\int_{\alpha}^{1}\sigma\left(p\right)\mathrm{d}p\quad\left(\alpha\in\left[0,1\right]\right)\mbox{ and }\int_{0}^{1}\tau\left(p\right)\mathrm{d}p=\int_{0}^{1}\sigma\left(p\right)\mathrm{d}p.
\]

\item The spectrum $\sigma$ majorizes the random variable $Z$ ($Z\preccurlyeq\sigma$)
iff 
\begin{align*}
(1-\alpha)\CTE_{\alpha}\left(Z\right)\le\int_{\alpha}^{1}\sigma\left(p\right)\mathrm{d}p & \quad\text{for all }\alpha\in\left[0,1\right]\\
\mbox{ and }\mathbb{E}Z=\int_{0}^{1}\sigma\left(p\right)\mathrm{d}p.
\end{align*}

\end{enumerate}
\end{defn}
\begin{rem}
Recall that for the conditional tail expectation it holds that 
\[
\left(1-\alpha\right)\CTE_{\alpha}\left(Z\right)=\int_{\alpha}^{1}F_{Z}^{-1}\left(p\right)\mathrm{d}p=\int_{\alpha}^{1}\tau\left(p\right)\mathrm{d}p,
\]
where $\tau$ is the function $\tau\left(\cdot\right):=F_{Z}^{-1}\left(\cdot\right)$.
It should thus be noted that 
\[
Z\preccurlyeq\sigma\mbox{ if and only if }F_{Z}^{-1}\preccurlyeq\sigma.
\]
Moreover $Z\preccurlyeq\sigma$ is related to a \emph{convex order
}or \emph{stochastic dominance} conditions, which are studied for
example in \cite{StoyanMueller2002} or \cite{shanked}. The dominance
in convex (concave) order was used in studying risk measures for example
in \cite{Follmer2004,Dana2005}. 
\end{rem}
The following Theorem~\ref{thm1} is a characterization of distorted
premium functionals by employing the convex conjugate relationship
for the dual. 

\bigskip{}

\begin{thm}[Representation of distorted premium functionals as a supremum by stochastic
order constraints.]
\label{thm1}Let $\pi_{\sigma}\left(L\right)$ be a distorted premium
functional.  Then the representation 
\begin{align}
\pi_{\sigma}\left(L\right) & =\sup\left\{ \mathbb{E}\, LZ\colon\: Z\preccurlyeq\sigma\right\} \label{eq:Dual}\\
 & =\sup\left\{ \mathbb{E}\, LZ\colon\:\mathbb{E}Z=1,\,\left(1-\alpha\right)\CTE_{\alpha}\left(Z\right)\le\int_{\alpha}^{1}\sigma\left(p\right)\mathrm{d}p,\,0\le\alpha<1\right\} \nonumber 
\end{align}
holds true. \end{thm}
\begin{rem}
The stochastic order constraint is employed here for the \emph{dual}
variable $Z$. Note also that the set $\left\{ Z\colon\, Z\preccurlyeq\sigma\right\} $
is closed, as 
\[
\left\{ Z\colon\, Z\preccurlyeq\sigma\right\} =\bigcap_{\begin{array}{c}
P\left(A\right)\le\alpha\\
0\le\alpha\le1
\end{array}}\left\{ Z\colon\,\mathbb{E}Z=1,\,\mathbb{E}\,\one_{A}Z\le\int_{\alpha}^{1}\sigma\left(p\right)\mathrm{d}p\right\} .
\]

\end{rem}

\begin{rem}
For the distortion $\sigma_{\mu}$ associated with $\mu$ (cf. \eqref{eq:3-1})
it holds that $\int_{\alpha}^{1}\sigma_{\mu}\left(p\right)\mathrm{d}p$
$=\int_{0}^{1}\min\left\{ \frac{1-\alpha}{1-p},1\right\} \mathrm{\mu\left(\mathrm{d}p\right)}$,
hence \eqref{eq:Dual} can be stated equivalently as 
\begin{align*}
\pi_{\sigma_{\mu}}\left(L\right) & =\sup\left\{ \mathbb{E}\, LZ\left|\begin{array}{l}
\mathbb{E}Z=1,\text{ and for all }\alpha\in\left(0,1\right)\\
\CTE_{\alpha}\left(Z\right)\le\int_{0}^{1}\min\left\{ \frac{1}{1-\alpha},\frac{1}{1-p}\right\} \mu\left(\mathrm{d}p\right)
\end{array}\right.\right\} 
\end{align*}
just by involving the measure $\mu$ from Kusuoka's representation.
\end{rem}

\begin{rem}
We emphasize that the conditions $(1-\alpha)\CTE_{\alpha}\left(Z\right)\le\int_{\alpha}^{1}\sigma\left(p\right)\mathrm{d}p$
and $\mathbb{E}Z=1$ together imply that $Z\ge0$ almost everywhere.
Indeed, suppose that $P\left(Z<0\right)=:p>0$. Then $1=\mathbb{E}Z=\int_{\left\{ Z<0\right\} }Z\mathrm{d}P+\int_{\left\{ Z\ge0\right\} }Z\mathrm{d}P=\int_{\left\{ Z<0\right\} }Z\mathrm{d}P+\left(1-p\right)\CTE_{p}(Z)$.
As $\int_{\left\{ Z<0\right\} }Z\mathrm{d}P<0$ it follows that $\left(1-p\right)$
$\CTE_{p}(Z)>1$. But this contradicts the fact that $(1-p)\CTE_{p}\left(Z\right)\le\int_{p}^{1}\sigma\left(p^{\prime}\right)\mathrm{d}p^{\prime}\le1$,
hence $Z$ is nonnegative, $Z\ge0$ almost surely.\end{rem}
\begin{proof}[Proof of Theorem~\ref{thm1}]
Recall the Legendre--Fenchel transformation for convex functions
(cf.~\cite{RuszczynskiShapiro2009}), 
\begin{eqnarray}
\pi_{\sigma}\left(L\right) & = & \sup_{Z}\mathbb{E}\, LZ-\pi_{\sigma}^{*}\left(Z\right),\text{ where}\nonumber \\
\pi_{\sigma}^{*}\left(Z\right) & = & \sup_{L}\mathbb{E}\, LZ-\pi_{\sigma}\left(L\right).\label{eq:Dual-2}
\end{eqnarray}
As $\pi_{\sigma}$ is version independent the random variable $L$
minimizing \eqref{eq:Dual-2} is coupled in a  co-monotone way with
$Z$ (cf.~\cite{Hoeffding} and \cite[Proposition 1.8]{PflugRomisch2007}
for the respective rearrangement inequality, sometimes referred to
as \emph{Hardy and Littlewood's inequality} or \emph{Hardy--Littlewood--Pólya
inequality}---cf.~\cite{Dana2005}). It follows that 
\begin{align*}
\pi_{\sigma}^{*}\left(Z\right) & =\sup_{L}\mathbb{E}LZ-\pi_{\sigma}\left(L\right)\\
 & =\sup\int_{0}^{1}F_{L}^{-1}\left(\alpha\right)F_{Z}^{-1}\left(\alpha\right)\mathrm{d}\alpha-\int_{0}^{1}F_{L}^{-1}\left(\alpha\right)\sigma\left(\alpha\right)\mathrm{d}\alpha,
\end{align*}
the infimum being among all cumulative distribution functions $F_{L}\left(y\right)=P\left(L\le y\right)$
of $L$. Define $G\left(\alpha\right):=\int_{\alpha}^{1}F_{Z}^{-1}\left(p\right)\mathrm{d}p$
and $S\left(\alpha\right):=\int_{\alpha}^{1}\sigma\left(p\right)\mathrm{d}p$,
whence 
\begin{align}
\pi_{\sigma}^{*}\left(Z\right) & =\sup_{F_{L}}\int_{0}^{1}F_{L}^{-1}\left(\alpha\right)\mathrm{d}\left(S\left(\alpha\right)-G\left(\alpha\right)\right)\nonumber \\
 & =\sup_{F_{L}}\left[F_{L}^{-1}\left(\alpha\right)\left(S\left(\alpha\right)-G\left(\alpha\right)\right)\right]_{\alpha=0}^{1}-\int_{0}^{1}S\left(\alpha\right)-G\left(\alpha\right)\mathrm{d}F_{L}^{-1}\left(\alpha\right)\nonumber \\
 & =\sup_{F_{L}}F_{L}^{-1}\left(0\right)\left(G\left(0\right)-S\left(0\right)\right)+\int_{0}^{1}G\left(\alpha\right)-S\left(\alpha\right)\mathrm{d}F_{L}^{-1}\left(\alpha\right)\label{eq:4}
\end{align}
by integration by parts of the Riemann--Stieltjes integral and as
it is enough to consider $L\in L^{\infty}$. 

Consider the constant random variables $L\equiv c$ ($c\in\mathbb{R}$),
then $F_{L}^{-1}\equiv c$ and, by \eqref{eq:4}, 
\[
\pi_{\sigma}^{*}\left(Z\right)\ge\sup_{c\in\mathbb{R}}\: c\left(G\left(0\right)-S\left(0\right)\right).
\]
Note now that $S\left(0\right)=\int_{0}^{1}\sigma\left(p\right)\mathrm{d}p=1$,
whence 
\[
\pi_{\sigma}^{*}\left(Z\right)\ge\sup_{c\in\mathbb{R}}\: c\left(G\left(0\right)-1\right)=\begin{cases}
0 & \mbox{if }G\left(0\right)=1\\
\infty & \mbox{else}
\end{cases}\quad=\begin{cases}
0 & \mbox{if }\mathbb{E}Z=1\\
\infty & \mbox{else},
\end{cases}
\]
because 
\begin{equation}
G\left(0\right)=\int_{0}^{1}F_{Z}^{-1}\left(p\right)\mathrm{d}p=\mathbb{E}\, Z.\label{eq:11}
\end{equation}
Assuming  $\mathbb{E}\, Z=1$ it follows from \eqref{eq:4} that 
\[
\pi_{\sigma}^{*}\left(Z\right)=\sup_{F_{L}}\int_{0}^{1}G\left(\alpha\right)-S\left(\alpha\right)\mathrm{d}F_{L}^{-1}\left(\alpha\right).
\]

Then choose an arbitrary measurable set $B$ and consider the random
variable $L_{c}:=c\cdot\one_{B^{\complement}}$ for some $c>0$. Note
that $F_{L_{c}}^{-1}=\one_{\left[\alpha_{0},1\right]}$, where $\alpha_{0}=P\left(B\right)$.
With this choice 
\begin{eqnarray*}
\pi_{\sigma}^{*}\left(Z\right) & \ge & \sup_{F_{L_{c}}}\int_{0}^{1}G\left(\alpha\right)-S\left(\alpha\right)\mathrm{d}F_{L_{c}}^{-1}\left(\alpha\right)\ge\sup_{c\ge0}c\left(G\left(\alpha_{0}\right)-S\left(\alpha_{0}\right)\right)=\\
 & = & \begin{cases}
0 & \mbox{if }G\left(\alpha{}_{0}\right)\le S\left(\alpha_{0}\right)\\
\infty & \mbox{else}
\end{cases}
\end{eqnarray*}
As $B$ was chosen arbitrarily it follows that $G\left(\alpha\right)\le S\left(\alpha\right)$
has to hold for any $0\le\alpha\le1$ for $Z$ to be feasible.

Conversely, if \eqref{eq:11} and $G\left(\alpha\right)\le S\left(\alpha\right)$
for all $0\le\alpha\le1$, then 
\[
\sup_{F_{L}}\int_{0}^{1}G\left(\alpha\right)-S\left(\alpha\right)\mathrm{d}F_{L}^{-1}\left(\alpha\right)\le0,
\]
because $F_{L}^{-1}\left(\cdot\right)$ is a nondecreasing function.
Note now that 
\begin{align*}
\int_{\alpha}^{1}\sigma\left(p\right)\mathrm{d}p & =S\left(\alpha\right)\ge G\left(\alpha\right)\\
 & =\int_{\alpha}^{1}F_{Z}^{-1}\left(p\right)\mathrm{d}p=\left(1-\alpha\right)\CTE_{\alpha}\left(Z\right),
\end{align*}
from which finally follows that 
\[
\pi_{\sigma}^{*}\left(Z\right)=\begin{cases}
0 & \mbox{if }\mathbb{E}Z=1\mbox{ and }\left(1-\alpha\right)\CTE_{\alpha}\left(Z\right)\le\int_{\alpha}^{1}\sigma\left(p\right)\mathrm{d}p\ \left(0\le\alpha\le1\right)\\
\infty & \mbox{else},
\end{cases}
\]
which is the assertion.
\end{proof}
The following statement derives naturally as a corollary of Theorem~\ref{thm1},
it will be essential in the sequel. 
\begin{cor}
Let $\pi_{\sigma}$ be a distortion risk functional, then
\begin{equation}
\pi_{\sigma}\left(L\right)=\sup\left\{ \mathbb{E}\, L\cdot\sigma\left(U\right)\colon\, U\mbox{ is uniformly distributed}\right\} ,\label{eq:RsigmaU}
\end{equation}
where the infimum is attained if $L$ and $U$ are coupled in a co-monotone
way. \end{cor}
\begin{rem}
The statement of the corollary implicitly and tacitly assumes that
the probability space is rich enough to carry a uniform random variable.
This is certainly the case if the probability space does not contain
atoms. But even if the probability space has atoms, then this is not
a restriction neither, as any probability space with atoms can be
augmented to allow a uniformly distributed random variable. \end{rem}
\begin{proof}
Consider $Z:=\sigma\left(U\right)$ for a uniformly distributed random
variable $U$, then $P\left(Z\le\sigma\left(\alpha\right)\right)=P\left(\sigma\left(U\right)\le\sigma\left(\alpha\right)\right)\ge P\left(U\le\alpha\right)=\alpha$,
that is $F_{Z}^{-1}\left(\alpha\right)\ge\sigma\left(\alpha\right)$.
But as $1=\int_{0}^{1}\sigma\left(\alpha\right)\mathrm{d}\alpha$
$\le\int_{0}^{1}F_{\sigma\left(U\right)}^{-1}\left(\alpha\right)\mathrm{d}\alpha$
$=\mathbb{E}\sigma\left(U\right)=1$ it follows that 
\[
F_{\sigma\left(U\right)}^{-1}\left(\cdot\right)=\sigma\left(\cdot\right)
\]
 almost everywhere. Observe now that any $Z$ with $F_{Z}^{-1}\left(\alpha\right)\le\sigma\left(\alpha\right)$
is feasible for \eqref{eq:Dual}, because 
\[
\int_{\alpha}^{1}\sigma\left(p\right)\mathrm{d}p\ge\int_{\alpha}^{1}F_{Z}^{-1}\left(p\right)\mathrm{d}p=(1-\alpha)\CTE_{\alpha}\left(Z\right)
\]
and $\mathbb{E}Z=\mathbb{E}\sigma\left(U\right)=\int_{0}^{1}\sigma\left(\alpha\right)\mathrm{d}\alpha=1$.
Now let $U$ be coupled in an co-monotone way with $L$, then $\mathbb{E}LZ=\int_{0}^{1}F_{L}^{-1}\left(\alpha\right)F_{Z}^{-1}\left(\alpha\right)\mathrm{d}\alpha=\int_{0}^{1}F_{L}^{-1}\left(\alpha\right)F_{\sigma\left(U\right)}^{-1}\left(\alpha\right)\mathrm{d}\alpha=\int_{0}^{1}F_{L}^{-1}\left(\alpha\right)\sigma\left(\alpha\right)\mathrm{d}\alpha$
such that
\[
\pi_{\sigma}\left(L\right)=\sup\left\{ \mathbb{E}\, L\sigma\left(U\right)\colon U\mbox{ uniformly distributed}\right\} ,
\]
which is finally the second assertion.
\end{proof}
The characterization derived in the previous theorem for spectral
premium functionals naturally applies to the conditional tail expectation
itself. The expression can be simplified further to give the dual
representation, which is often used to define the conditional tail
expectation. The second statement exhibits an interesting, ``recursive''
structure. 
\begin{cor}
\label{cor:AVaR}The conditional tail expectation at level $\alpha$
obeys the dual representations 
\begin{align*}
\CTE_{\alpha}\left(L\right) & =\sup\left\{ \mathbb{E}LZ:\,\mathbb{E}Z=1,\,0\le Z,\,(1-\alpha)Z\le\one\right\} \\
 & =\sup\left\{ \mathbb{E}LZ:\,\mathbb{E}Z=1,\,\CTE_{p}\left(Z\right)\le\frac{1}{1-\alpha}\mbox{ for all }p>\alpha\right\} .
\end{align*}
\end{cor}
\begin{proof}
The conditional tail expectation at level $\alpha$ is provided by
the Dirac measure $\mu_{\alpha}\left(A\right):=\delta_{\alpha}\left(A\right)=\begin{cases}
1 & \mbox{if }\alpha\in A\\
0 & \text{otherwise}
\end{cases}$, and the respective distortion function is $\sigma_{\alpha}$ (cf.~\eqref{eq:sigma}).
It follows from $\int_{p}^{1}\sigma_{\alpha}\left(p^{\prime}\right)\mathrm{d}p^{\prime}=\min\left\{ 1,\,\frac{1-p}{1-\alpha}\right\} $
and Theorem~\ref{thm1} that 

\[
\CTE_{\alpha}\left(L\right)=\inf\left\{ \mathbb{E}LZ:\,\mathbb{E}Z=1,\,\CTE_{p}\left(Z\right)\le\min\left\{ \frac{1}{1-p},\frac{1}{1-\alpha}\right\} \right\} .
\]
Observe next that for $Z\ge0$ 
\begin{align*}
\frac{1}{1-p}=\frac{1}{1-p}\mathbb{E}Z & \ge\frac{1}{1-p}\int_{p}^{1}F_{Z}^{-1}\left(p^{\prime}\right)\mathrm{d}p^{\prime}=\CTE_{p}\left(Z\right),
\end{align*}
hence
\[
\CTE_{\alpha}\left(L\right)=\inf\left\{ \mathbb{E}LZ:\,\mathbb{E}Z=1,\,\CTE_{p}\left(Z\right)\le\frac{1}{1-\alpha}\right\} .
\]
For $p\le\alpha$, in addition, $\CTE_{p}\left(Z\right)\le\frac{1}{1-p}\le\frac{1}{1-\alpha}$. 

This proves the second assertion.

As for the first observe that $\frac{1}{1-\alpha}\ge\CTE_{p}\left(Z\right)\rightarrow\esssup Z$,
hence $\left(1-\alpha\right)Z\le\one$; conversely, if $0\le Z$ and
$(1-\alpha)Z\le\one$, then 
\[
\frac{1}{1-\alpha}\ge\esssup Z\ge\CTE_{p}\left(Z\right),
\]
which is the first assertion.
\end{proof}

\section{\label{sec:Infimum-Representation}Infimum Representation Of Distortion
Pre\-mium Functionals }

The latter Theorem~\ref{thm1} exposes the distorted risk premium
as a supremum and characterizes the convex conjugate function by stochastic
dominance constraints. The following theorem, the second main result
of this article, provides a description in opposite terms, as an infimum.
The representation extends the well known formula for the conditional
tail expectation (Average Value-at-Risk) provided in \cite{RockafellarUryasev2000},
finally stated in the present form in \cite{Pflug2000}. 

This alternative description allows an alternative view on distortions
and alternative simulations, as is the content of the following section. 
\begin{thm}[Representation as an Infimum]
\label{thm:InfRep}For any $L\in L^{\infty}$ the distorted premium
functional with distortion $\sigma$ has the representation 
\begin{equation}
\pi_{\sigma}\left(L\right)=\inf_{h}\,\mathbb{E}\, h\left(L\right)+\int_{0}^{1}h^{*}\left(\sigma\left(p\right)\right)\mathrm{d}p,\label{eq:HS}
\end{equation}
where the infimum is among all arbitrary, measurable functions $h\colon\mathbb{R}\rightarrow\mathbb{R}$
and $h^{*}$ is $h$'s convex conjugate function %
\footnote{The convex conjugate function of $h$  is $h^{*}\left(y\right):=\sup_{x}\, x\cdot y-h\left(x\right)$.
The convex conjugate may evaluate to $+\infty$.%
}.\end{thm}
\begin{rem}
Having a look at representation \eqref{eq:HS} it is not immediate
that the axioms of Definition~\ref{Def-RiskMeasure} are satisfied.
The transformations listed in Lemma~\ref{lem:Transform} in the Appendix
can be used in a straight forward manner to deduce the properties
directly from \eqref{eq:HS}. 
\end{rem}
The statement of the Inf-Representation Theorem~\ref{thm:InfRep}
can be formulated equivalently in the following ways.
\begin{cor}
\label{cor:5}For any $L\in L^{\infty}$ the distorted risk premium
with distortion $\sigma$ allows the representations 
\begin{eqnarray}
\pi_{\sigma}\left(L\right) & = & \inf_{f\text{ convex}}\,\mathbb{E}\, h\left(L\right)+\int_{0}^{1}h^{*}\left(\sigma\left(p\right)\right)\mathrm{d}p\nonumber \\
 & = & \inf\left\{ \mathbb{E}\, h\left(L\right):\int_{0}^{1}h^{*}\left(\sigma\left(p\right)\right)\mathrm{d}p\le0\right\} ,\label{eq:hSigma}
\end{eqnarray}
where the latter infimum is among arbitrary, measurable functions
$h\colon\mathbb{R}\rightarrow\mathbb{R}$.\end{cor}
\begin{proof}[Proof of Corollary~\ref{cor:5}]
It is well known that the bi-conjugate function $h^{**}:=\left(h^{*}\right)^{*}$
is a convex and lower semicontinuous function satisfying $h^{**}\le h$
and $h^{***}=h^{*}$ (cf. the analogous Fenchel--Moreau Theorem and
equation \eqref{eq:Dual-2}). The infimum in \eqref{eq:HS} hence---without
any loss of generality---can be restricted to \emph{convex} functions,
that is 
\[
\pi_{\sigma}\left(L\right)=\inf_{h\text{ convex}}\,\mathbb{E}h\left(L\right)+\int_{0}^{1}h^{*}\left(\sigma\left(p\right)\right)\mathrm{d}p.
\]

As for the second assertion notice first that  clearly 
\begin{eqnarray*}
\pi_{\sigma}\left(L\right) & \le & \inf\left\{ \mathbb{E}h\left(L\right)+\int_{0}^{1}h^{*}\left(\sigma\left(p\right)\right)\mathrm{d}p:\int_{0}^{1}h^{*}\left(\sigma\left(p\right)\right)\mathrm{d}p\le0\right\} \\
 & \le & \inf\left\{ \mathbb{E}h\left(L\right):\int_{0}^{1}h^{*}\left(\sigma\left(p\right)\right)\mathrm{d}p\le0\right\} .
\end{eqnarray*}
Consider $h_{\alpha}\left(x\right):=h(x)-\alpha$ (where $\alpha$
a constant and $h$ arbitrary). It holds that $h_{\alpha}^{*}\left(y\right)=h^{*}(y)+\alpha$,
as exposed by the auxiliary Lemma~\ref{lem:Transform} in the Appendix.
Hence $\int_{0}^{1}h_{\alpha}^{*}\left(\sigma\left(p\right)\right)\mathrm{d}p=\int_{0}^{1}h^{*}\left(\sigma\left(p\right)\right)\mathrm{d}p+\alpha$
and 
\begin{equation}
\mathbb{E}\, h_{\alpha}\left(L\right)+\int_{0}^{1}h_{\alpha}^{*}\left(\sigma\left(p\right)\right)\mathrm{d}p=\mathbb{E}\, h\left(L\right)+\int_{0}^{1}h^{*}\left(\sigma\left(p\right)\right)\mathrm{d}p.\label{eq:12}
\end{equation}
Choose $\alpha:=\int_{0}^{1}h^{*}\left(\sigma\left(p\right)\right)\mathrm{d}p$
such that $\int_{0}^{1}h_{\alpha}^{*}\left(\sigma\left(p\right)\right)\mathrm{d}p=0$.
$h_{\alpha}$ hence is feasible for \eqref{eq:hSigma} with the same
objective as $h$ by \eqref{eq:12}, from which the assertion follows.
\end{proof}

\begin{rem}
Notice that $\sigma$ has its range in the interval $\left\{ \sigma\left(x\right):\, x\in\left[0,1\right]\right\} =\left[0,\sigma\left(1\right)\right]$,
and from convexity of $h^{*}$ it follows that the set $\left\{ h^{*}<\infty\right\} $
is convex. Hence $h^{*}\left(y\right)<\infty$ necessarily has to
hold for all $y\in\left(0,\sigma\left(1\right)\right)$ to ensure
that $\int_{0}^{1}h^{*}\left(\sigma(u)\right)\mathrm{d}u<\infty$.
For $h$ convex this means in turn that 
\[
\lim_{x\to-\infty}h^{\prime}\left(x\right)\le0\text{ and }\lim_{x\to\infty}h^{\prime}\left(x\right)\ge\sigma\left(1\right),
\]
limiting thus the class of interesting functions in Corollary~\ref{cor:5}
to convex functions satisfying $h^{\prime}\left(\mathbb{R}\right)\supset\left(0,\sigma(1)\right)$. \end{rem}
\begin{proof}[Proof of Theorem \ref{thm:InfRep}]
 From the definition of the convex conjugate $h^{*}$ it is immediate
that 
\[
h^{*}\left(\sigma\right)\ge y\cdot\sigma-h\left(y\right)
\]
for all numbers $y$ and $\sigma$ (this is often called \emph{Fenchel--Young
inequality}), hence 
\[
h\left(L\right)+h^{*}\left(\sigma\left(U\right)\right)\ge L\cdot\sigma\left(U\right),
\]
where $U$ is any uniformly distributed random variable, i.e.\ $U$
satisfies $P\left(U\le u\right)=u$. Taking expectations it follows
that 
\[
\mathbb{E}h\left(L\right)+\mathbb{E}h^{*}\left(\sigma\left(U\right)\right)\ge\mathbb{E}\, L\cdot\sigma\left(U\right).
\]
As $U$ is uniformly distributed it holds that 
\[
\mathbb{E}h^{*}\left(\sigma\left(U\right)\right)=\int_{0}^{1}h^{*}\left(\sigma\left(u\right)\right)\mathrm{d}u,
\]
such that 
\[
\mathbb{E}h\left(L\right)+\int_{0}^{1}h^{*}\left(\sigma\left(u\right)\right)\mathrm{d}u\ge\mathbb{E}\, L\cdot\sigma\left(U\right),
\]
irrespective of the uniform random variable $U$. Hence, by $\eqref{eq:RsigmaU}$
in Corollary~\ref{thm1}, 
\[
\mathbb{E}h\left(L\right)+\int_{0}^{1}h^{*}\left(\sigma\left(u\right)\right)\mathrm{d}u\ge\sup_{U\text{ uniform}}\mathbb{E}\, L\cdot\sigma\left(U\right)=\pi_{\sigma}\left(L\right),
\]
 establishing the inequality 
\[
\pi_{\sigma}\left(L\right)\le\mathbb{E}h\left(L\right)+\int_{0}^{1}h^{*}\left(\sigma\left(u\right)\right)\mathrm{d}u.
\]

As for the converse inequality consider the function 
\begin{equation}
h_{\sigma}\left(y\right):=\int_{0}^{1}F_{L}^{-1}\left(\alpha\right)+\frac{1}{1-\alpha}\left(y-F_{L}^{-1}\left(\alpha\right)\right)_{+}\mu_{\sigma}\left(\mathrm{d}\alpha\right).\label{eq:Optimum}
\end{equation}
$h_{\sigma}\left(y\right)$ is well defined for all $y$ because $L\in L^{\infty}$;
$h_{\sigma}\left(y\right)$ is moreover increasing and convex, because
$y\mapsto\left(y-q\right)_{+}$ is increasing and convex, and because
$\mu_{\sigma}$ is positive. 

Recall the formula 
\[
\CTE_{\alpha}\left(L\right)=\inf_{q\in\mathbb{R}}\: q+\frac{1}{1-\alpha}\mathbb{E}\left(L-q\right)_{+}
\]
and the fact that the infimum is attained at $q=F_{L}^{-1}\left(\alpha\right)$
(cf.~\cite{Pflug2000} or \cite[Section 4.1]{Goovaerts2012} for
the general  formula), providing thus the explicit form 
\[
\CTE_{\alpha}\left(L\right)=F_{L}^{-1}\left(\alpha\right)+\frac{1}{1-\alpha}\mathbb{E}\left(L-F_{L}^{-1}\left(\alpha\right)\right)_{+}.
\]
Note now that, by Fubini's Theorem, 
\begin{align}
\pi_{\sigma}\left(L\right) & =\int_{0}^{1}\CTE_{\alpha}\left(L\right)\mu_{\sigma}\left(\mathrm{d}\alpha\right)\nonumber \\
 & =\int_{0}^{1}F_{L}^{-1}\left(\alpha\right)+\frac{1}{1-\alpha}\mathbb{E}\left(L-F_{L}^{-1}\left(\alpha\right)\right)_{+}\mu_{\sigma}\left(\mathrm{d}\alpha\right)\nonumber \\
 & =\mathbb{E}\int_{0}^{1}\, F_{L}^{-1}\left(\alpha\right)+\frac{1}{1-\alpha}\left(L-F_{L}^{-1}\left(\alpha\right)\right)_{+}\mu_{\sigma}\left(\mathrm{d}\alpha\right)\nonumber \\
 & =\mathbb{E}\, h_{\sigma}\left(L\right).\label{eq:f0}
\end{align}

To establish the assertion \eqref{eq:HS} it needs to be shown that
$\int_{0}^{1}h_{\sigma}^{*}\left(\sigma\left(u\right)\right)\mathrm{d}u\le0$.
For this observe first that $h_{\sigma}$ is almost everywhere differentiable
(because it is convex), with derivative  
\begin{eqnarray}
h_{\sigma}^{\prime}\left(y\right) & = & \int_{\left\{ \alpha\colon F_{L}^{-1}\left(\alpha\right)\le y\right\} }\frac{1}{1-\alpha}\mu_{\sigma}\left(\mathrm{d}\alpha\right)\nonumber \\
 & = & \int_{0}^{F_{L}(y)}\frac{1}{1-\alpha}\mu_{\sigma}\left(\mathrm{d}\alpha\right)=\sigma\left(F_{L}(y)\right)\label{eq:18}
\end{eqnarray}
(almost everywhere) by relation \eqref{eq:3-1}. Moreover $h_{\sigma}^{*}\left(\sigma(u)\right)=\sup_{y}\sigma(u)\cdot y-h_{\sigma}\left(y\right)$,
the supremum being attained at every $y$ satisfying $\sigma(u)=h_{\sigma}^{\prime}\left(y\right)=\sigma\left(F_{L}(y)\right)$,
hence at $y=F_{L}^{-1}\left(u\right)$, and it follows that  
\[
h_{\sigma}^{*}\left(\sigma\left(u\right)\right)=\sigma\left(u\right)\cdot F_{L}^{-1}\left(u\right)-h_{\sigma}\left(F_{L}^{-1}\left(u\right)\right).
\]
Now 
\begin{eqnarray*}
\int_{0}^{1}h_{\sigma}^{*}\left(\sigma\left(u\right)\right)\mathrm{d}u & = & \int_{0}^{1}\sigma\left(u\right)\cdot F_{L}^{-1}\left(u\right)\mathrm{d}u-\int_{0}^{1}h_{\sigma}\left(F_{L}^{-1}\left(u\right)\right)\mathrm{d}u\\
 & = & \pi_{\sigma}\left(L\right)-\mathbb{E}h_{\sigma}\left(L\right).
\end{eqnarray*}
But it was established already in \eqref{eq:f0} that $\pi_{\sigma}\left(L\right)=\mathbb{E}h_{\sigma}\left(L\right)$,
so that 
\[
\int_{0}^{1}h_{\sigma}^{*}\left(\sigma\left(u\right)\right)\mathrm{d}u=0.
\]
This finally proves the second inequality. 
\end{proof}

\paragraph*{$\CTE$ as a special case.}

The conditional tail expectation is a special case of the infimum
in \eqref{eq:HS}. Indeed, it follows from \eqref{eq:Optimum} in
the proof that the infimum is attained at a function of the form $h_{q}\left(y\right)=q+\frac{1}{1-\alpha}\left(y-q\right)_{+}$
with conjugate 
\[
h_{q}^{*}\left(x\right)=\begin{cases}
-q+q\, x & \text{ if }0\le x\le\frac{1}{1-\alpha}\\
\infty & \text{ else}.
\end{cases}
\]
It holds that 
\begin{eqnarray*}
\int_{0}^{1}h_{\sigma}^{*}\left(\sigma_{\alpha}\left(x\right)\right)\mathrm{d}x & = & \int_{0}^{\alpha}h_{\sigma}^{*}\left(0\right)\mathrm{d}x+\int_{\alpha}^{1}h_{\sigma}^{*}\left(\frac{1}{1-\alpha}\right)\mathrm{d}x\\
 & = & -\alpha q+\left(-q+\frac{q}{1-\alpha}\right)(1-\alpha)=0,
\end{eqnarray*}
such that 
\begin{equation}
\CTE_{\alpha}\left(L\right)=\inf_{q\in\mathbb{R}}\:\mathbb{E}\, h_{q}\left(L\right)=\inf_{q}\: q+\frac{1}{1-\alpha}\mathbb{E}\left(L-q\right)_{+},\label{eq:AVaR}
\end{equation}
the classical result. Clearly, the infimum in \eqref{eq:AVaR} is
in $\mathbb{R}$, a much smaller space than convex functions from
$\mathbb{R}$ to $\mathbb{R}$, as required in \eqref{eq:HS}.

\section{\label{sec:StochasticOptimization}Implications for Actuarial Science
and Claim Sampling}

\subsection{Comparison of $L_{\sigma}$ and $L_{\sigma}^{\prime}$}

\begin{figure}
\subfloat[The function $h_{\sigma}$ of the normal distribution for distortion
function $\sigma(u)=0.7+0.9u^{2}$.]{\includegraphics[width=0.45\textwidth]{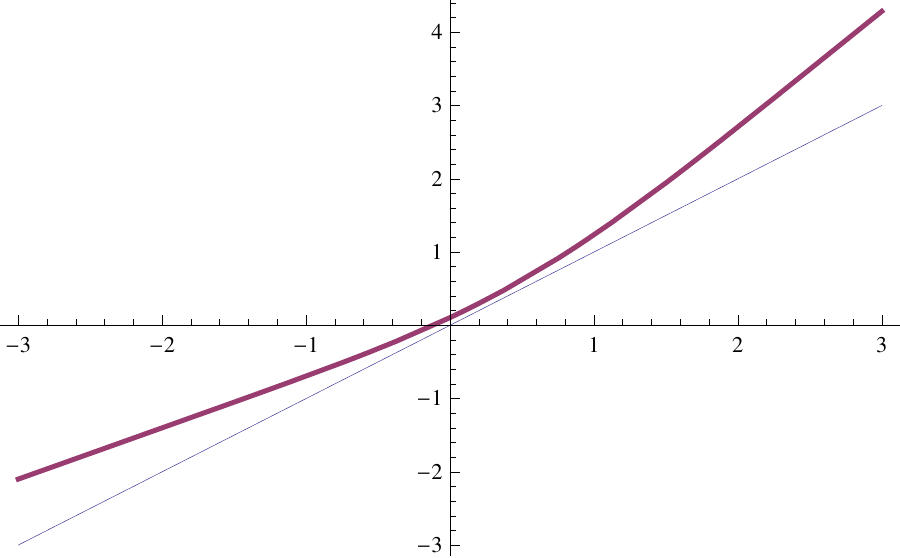}

}$\qquad$\subfloat[The density of the normal distribution, together with the density
of $L_{\sigma}$ (right mode) and density of $L_{\sigma}^{\prime}$
(left mode).]{\includegraphics[width=0.49\textwidth]{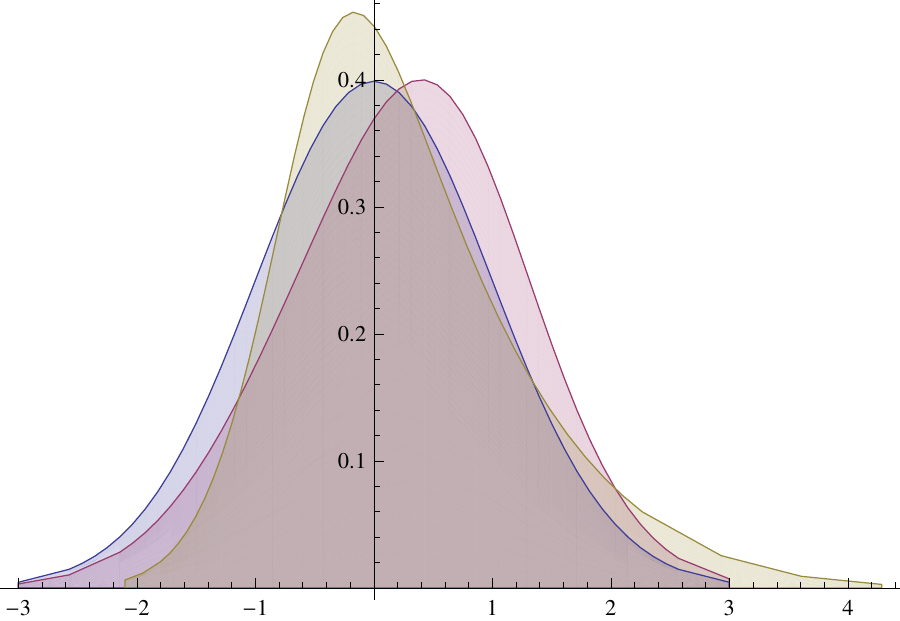}

}

\caption{\label{fig:1}Distortion of the standard normal distribution.}
\end{figure}
In the introductory discussion it was outlined that claims can be
sampled (based on \eqref{eq:2-1}) by use of 
\[
L_{\sigma}=F_{L}^{-1}\left(\tau_{\sigma}^{-1}\left(U\right)\right).
\]
It is obvious by this formula that the distorted claims $L_{\sigma}$
have the same outcomes as $L$, but their probability is disturbed
by involvement of the function $\tau_{\sigma}$. 

The infimum representation developed in Section~\ref{sec:Infimum-Representation}
suggests to consider the random variable 
\[
L_{\sigma}^{\prime}:=h_{\sigma}\left(L\right),
\]
where $h_{\sigma}$ is the function defined in \eqref{eq:Optimum},
and which is the optimal function for problem \eqref{eq:hSigma}.
For this function it holds that
\[
\mathbb{E}\, L_{\sigma}^{\prime}=\pi_{\sigma}\left(L\right)=\mathbb{E}\, L_{\sigma},
\]
because $\int_{0}^{1}h_{\sigma}^{*}\left(\sigma\left(u\right)\right)\mathrm{d}u=0$.
We have moreover that 
\[
h_{\sigma}\left(y\right)=\int_{0}^{1}F_{L}^{-1}\left(\alpha\right)+\frac{1}{1-\alpha}\left(y-F_{L}^{-1}\left(\alpha\right)\right)_{+}\mu_{\sigma}\left(\mathrm{d}\alpha\right)\ge\int_{0}^{1}y\,\mu_{\sigma}\left(\mathrm{d}\alpha\right)=y,
\]
from which follows that 
\[
L_{\sigma}^{\prime}\ge L,
\]
that is, $L_{\sigma}^{\prime}$ stochastically dominates $L$ in first
order. The cumulative distribution function of $L_{\sigma}^{\prime}$
has the explicit form 
\[
F_{L_{\sigma}^{\prime}}\left(y\right)=P\left(h_{\sigma}\left(L\right)\le y\right)=P\left(L\le h_{\sigma}^{-1}\left(y\right)\right)=F_{L}\left(h_{\sigma}^{-1}\left(y\right)\right),
\]
and the density is $f_{L_{\sigma}^{\prime}}\left(y\right)=\frac{f_{L}\left(h_{\sigma}^{-1}\left(y\right)\right)}{h_{\sigma}^{\prime}\left(h_{\sigma}^{-1}\left(y\right)\right)}=\frac{f_{L}\left(h_{\sigma}^{-1}\left(y\right)\right)}{\sigma\left(F_{L}\left(h_{\sigma}^{-1}\left(y\right)\right)\right)}$
by use of \eqref{eq:18}. The quantile function 
\begin{equation}
F_{L_{\sigma}^{\prime}}^{-1}=h_{\sigma}\circ F_{L}^{-1}.\label{eq:17}
\end{equation}
 is obtained by inversion. 
\begin{example}
Figure~\ref{fig:1} contains the densities of both distortions, $L_{\sigma}$
and $L_{\sigma}^{\prime}$, for the standard normal distribution.
The distortion function chosen in this example is $\sigma\left(u\right)=0.7+0.9u^{2}$.
This example reveals that the mode, as well as the tails of the random
variables $L_{\sigma}$ and $L_{\sigma}^{\prime}$ differ significantly;
the tails of $L_{\sigma}^{\prime}$ are heavier.
\end{example}

\paragraph{Opposite perspectives.}

The latter formula \eqref{eq:17} reveals that $L_{\sigma}^{\prime}$
has distorted outcomes, distorted by $h_{\sigma}$, but the probabilities
are unchanged. So $L_{\sigma}^{\prime}$ can be considered as alternative
to \eqref{eq:3}, doing exactly the opposite of the formula~\eqref{eq:3}
stated initially: $L_{\sigma}^{\prime}$ has the same probabilities
as $L$, but the outcomes are distorted by $h_{\sigma}$ whereas $L_{\sigma}$
has the same outcomes as $L$, but the probabilities are distorted
by $\tau_{\sigma}$. However, both, $L_{\sigma}$ and $L_{\sigma}^{\prime}$,
have the same expected value 
\[
\mathbb{E}\, L_{\sigma}=\pi_{\sigma}\left(L\right)=\mathbb{E}\, L_{\sigma}^{\prime}.
\]

\paragraph{Explicit distances. }

As the cumulative distribution function is available for $L_{\sigma}$
and $L_{\sigma}^{\prime}$ as elaborated, explicit expressions are
available for selected distances of random variables. An explicit
representation of the Kolmogorov--Smirnov distance for example is
\[
\sup_{y\in\mathbb{R}}\left|F_{L}\left(y\right)-\tau_{\sigma}\left(F_{L}\left(h_{\sigma}\left(y\right)\right)\right)\right|,
\]
and the Wasserstein distance (cf.~\cite{Villani2003}) has the explicit
formula 
\[
\int_{0}^{1}\left|h_{\sigma}\left(F_{L}^{-1}\left(\tau_{\sigma}\left(y\right)\right)\right)-F_{L}^{-1}\left(y\right)\right|\sigma\left(u\right)\mathrm{d}u.
\]

\subsection{Actuarial Applications}

\begin{figure}[t]
\includegraphics[width=\textwidth]{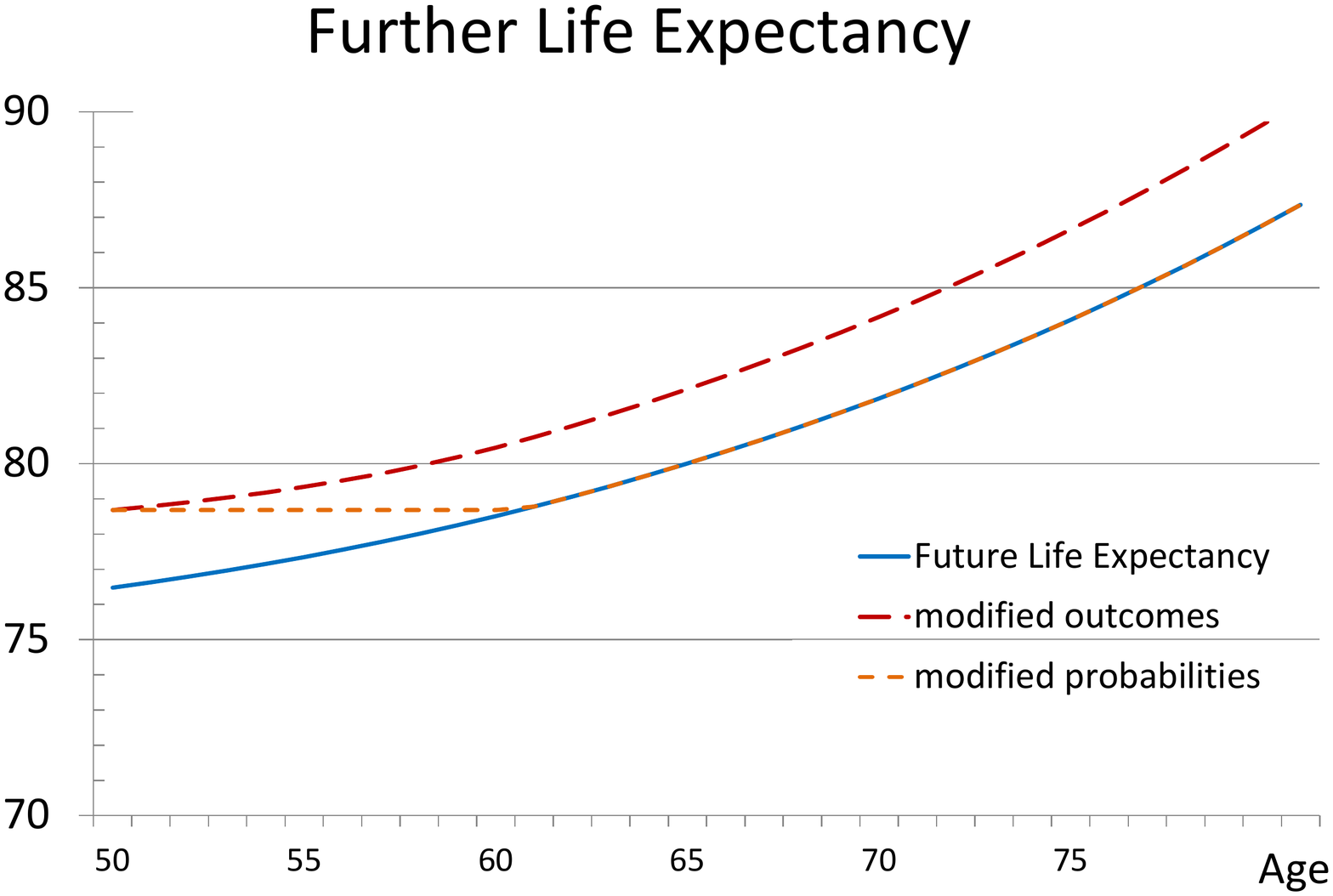}

\caption{\label{fig:Expectancies}Further life expectancies based on distorted
outcomes, and based on distorted probabilities. The distortion employed
is the conditional tail expectation at the level of $10\,\%$, $\CTE_{10\%}$.}
\end{figure}

Actuarial concerns have been addressed on various locations of the
paper, however, we stress again that $\pi$, $\pi_{\sigma}$ and in
particular $\CTE$ constitute premium principles. For a given loss
distribution with monotone (increasing, or decreasing) loss function
$L$ (note, that this is almost always the case in life insurance),
the function $\pi_{\sigma}\left(L\right)=\int_{0}^{1}F_{L}^{-1}\left(u\right)\sigma\left(u\right)\mathrm{d}u$
can be given in a closed form.
\begin{example}
Considering the simple life expectancy,%
\footnote{Note that the life expectancy is an annuity with an interest rate
of $0\,\%$. We have chosen an annuity as a representative example
for a typical life insurance contract. Considering the life expectancy
allows moreover excluding the interest rate in order to simplify the
presented results.%
} i.e. the random variable $L\left(k\right)=k$ (which is strictly
increasing), then $F_{L}^{-1}\left(_{k}q_{x}\right)=k$ and 
\[
\pi_{\sigma}\left(L\right)=\int_{0}^{1}F_{L}^{-1}\left(u\right)\sigma\left(u\right)\mathrm{d}u=\sum_{k=0}k\cdot\int_{_{k}q_{x}}^{_{k+1}q_{x}}\sigma\left(u\right)\mathrm{d}u
\]
is the distorted life expectancy. 
\end{example}

\paragraph{Distorted probabilities.}

Following \eqref{eq:2-2} one may consider $\int_{_{k}q_{x}}^{_{k+1}q_{x}}\sigma\left(u\right)=:{}_{k}\!\tilde{p}_{x}\cdot\tilde{q}_{x+k}$
as probability of a new life table (indicated by the tilde), and use
this new life table to compute premiums, as well as reserves. This
is exemplary depicted in Figure~\ref{fig:Expectancies}. It is visible
in this chart that the modified life table increases the life expectancy
by approximately $2$ years initially, but the increasing effect disappears
at the age representing the quantile (here, at the age of 60 years
for $\alpha=10\,\%$, considering a person with an initial age of
50). For this reason it is appropriate to use $\pi_{\sigma}\left(L\right)$
as a premium, but it is \emph{not }desirable to use the new life table
to compute reserves. The reserves loose the safety loading by employing
the new life table, whenever the age exceeds the quantile.

\paragraph{Distorted outcomes.}

As already outlined it is natural to use the distorted outcomes instead
of distorted probabilities in actuarial practice. As to compute the
premiums the above discussion applies equally well, and an explicit
form is available to compute the premium. For the exposed case of
life expectancy the result is 
\[
\pi_{\sigma}\left(L\right)=\int_{0}^{1}F_{L}^{-1}\left(u\right)\sigma\left(u\right)\mathrm{d}u=\sum_{k=0}h_{\sigma}\left(k\right)\cdot{}_{k}p_{x}\: q_{x+k}.
\]

It is the big advantage of distorted outcomes, that the reserves can
be handled with the same ingredients as the premium, that is with
the same probabilities and the same function $h_{\sigma}$: $L$ simply
needs to be replaced by $L_{\sigma}^{\prime}=h_{\sigma}\left(L\right)$.
It is evident in Figure~\ref{fig:Expectancies} that the safety loading
is preserved over time. 

Distorted premiums, interpreted as distorted outcomes, are thus a
reliable premium principle which provide not only premiums, but also
reserves in a correct and time-consistent way. The distorted premium
principle $\pi_{\sigma}$ to compute the reserves can be applied by
the actuary easily, and along with the related outcomes distorted
by $h_{\sigma}$.

\section{Concluding Remarks}

This article outlines new descriptions of distorted premium principles.
Distorted premium principles constitute a basic class of premium principles,
as every premium satisfying sufficiently strong axioms can be built
by involving just elementary distortions. 

The first representation derived is described as a supremum, based
on conjugate duality. The convex conjugate function is formulated
in terms of second order stochastic dominance constrains.

The other representation, which is a further central result of this
article, is described as an infimum and can be considered as the opposite
formulation. This alternative description makes distorted premiums
eligible for successful use in actuarial applications, as the reserve
process is easily available for concrete insurance contracts and,
above all, the process of reserves is consistent over time. The results
thus make distorted premiums eligible for extended actuarial use.

\section{Acknowledgment}

The charts have been programmed by use of Mathematica and Excel.

\bibliographystyle{alpha}
\bibliography{../../Literatur/LiteraturAlois}

\begin{thebibliography}{GLWT12}

\bibitem[Ace02]{Acerbi2002a}
Carlo Acerbi.
\newblock Spectral measures of risk: A coherent representation of subjective
  risk aversion.
\newblock {\em Journal of Banking \& Finance}, 26:1505--1518, 2002.

\bibitem[ADEH99]{Artzner1999}
Philippe Artzner, Freddy Delbaen, Jean-Marc Eber, and David Heath.
\newblock Coherent {M}easures of {R}isk.
\newblock {\em Mathematical Finance}, 9:203--228, 1999.

\bibitem[ADH97]{Artzner1997}
Philippe Artzner, Freddy Delbaen, and David Heath.
\newblock Thinking coherently.
\newblock {\em Risk}, 10:68--71, November 1997.

\bibitem[AS02]{Acerbi2002}
Carlo Acerbi and Prospero Simonetti.
\newblock Portfolio optimization with spectral measures of risk.
\newblock {\em EconPapers}, 2002.

\bibitem[BBH09]{Balbas2009a}
Alejandro Balbás, Beatriz Balbás, and Antonio Heras.
\newblock Optimal reinsurance with general risk measures.
\newblock {\em Insurance: Mathematics and Economics}, 44:374--384, 2009.

\bibitem[CT11]{Chi2011}
Yichun Chi and Ken~Seng Tan.
\newblock Optimal reinsurance under {V}a{R} and {CV}a{R} risk measures.
\newblock {\em ASTIN Bulletin}, 41(2):487--509, 2011.

\bibitem[Dan05]{Dana2005}
Rose-Anne Dana.
\newblock A representation result for concave {S}chur concave functions.
\newblock {\em Mathematical Finance}, 15:613--634, 2005.

\bibitem[Den90]{Denneberg1989}
Dieter Denneberg.
\newblock Distorted probabilities and insurance premiums.
\newblock {\em Methods of Operations Research}, 63:21--42, 1990.

\bibitem[FS04]{Follmer2004}
Hans F\"ollmer and Alexander Schied.
\newblock {\em Stochastic Finance: An Introduction in Discrete Time}.
\newblock de Gruyter Studies in Mathematics 27. de Gruyter, 2004.

\bibitem[FZ08]{Furman2008}
Edward Furman and Ri\v{c}ardas Zitikis.
\newblock Weighted premium calculation principles.
\newblock {\em Insurance: Mathematics and Economics}, 42:459--465, 2008.

\bibitem[Gio05]{Giorgi2005}
Enrico~De Giorgi.
\newblock Reward-risk portfolio selection and stochastic dominance.
\newblock {\em Journal of Banking \& Finance}, 29:895--926, 2005.

\bibitem[GLWT12]{Goovaerts2012}
Marc Goovaerts, Dani\"{e}l Linders, Koen~Van Weert, and Fatih Tank.
\newblock On the interplay between distortion, mean value and the
  {H}aezendonck-{G}oovaerts risk measures.
\newblock {\em Insurance: Mathematics and Economics}, 51:10--18, 2012.

\bibitem[HBV12]{Balbas2012}
Antonio Heras, Beatriz Balbás, and José~Luis Vilar.
\newblock Conditional tail expectation and premium calculation.
\newblock {\em ASTIN Bulletin}, 42:325--342, 2012.

\bibitem[Hoe40]{Hoeffding}
Wassilij Hoeffding.
\newblock Maßstabinvariante {K}orrelationstheorie.
\newblock {\em Schriften Math. Inst. Univ. Berlin}, 5:181--233, 1940.
\newblock In German.

\bibitem[JST06]{Schachermayer2006}
Elyès Jouini, Walter Schachermayer, and Nizar Touzi.
\newblock Law invariant risk measures have the {F}atou property.
\newblock {\em Advances in Mathematical Economics}, 9:49--71, 2006.

\bibitem[Kro07]{Krokhmal2007}
Pavlo~A. Krokhmal.
\newblock Higher moment coherent risk measures.
\newblock {\em Quantitative Finance}, 7(4):373--387, 2007.

\bibitem[KRS09]{Bangwon2009}
Bangwon Ko, Ralph~P. Russo, and Nariankadu~D. Shyamalkumar.
\newblock A note on the nonparametric estimation of the {CTE}.
\newblock {\em ASTIN Bulletin}, 39(2):717--734, 2009.

\bibitem[Kus01]{Kusuoka}
Shigeo Kusuoka.
\newblock On law invariant coherent risk measures.
\newblock {\em Advances in Mathematical Economics}, 3:83--95, 2001.

\bibitem[MS02]{StoyanMueller2002}
Alfred M\"uller and Dietrich Stoyan.
\newblock {\em Comparison methods for stochastic models and risks}.
\newblock Wiley series in probability and statistics. Wiley, Chichester, 2002.

\bibitem[Pfl00]{Pflug2000}
Georg~{\relax Ch}. Pflug.
\newblock Some remarks on the value-at-risk and the conditional value-at-risk.
\newblock In S.~Uryasev, editor, {\em Probabilistic Constrained Optimization:
  Methodology and Applications}, pages 272--281. Kluwer Academic Publishers,
  Dordrecht, 2000.

\bibitem[Pfl06]{Pflug2006}
Georg~{\relax Ch}. Pflug.
\newblock On distortion functionals.
\newblock {\em Statistics and Risk Modeling (formerly: Statistics and
  Decisions)}, 24:45--60, 2006.

\bibitem[PR07]{PflugRomisch2007}
Georg~{\relax Ch}. Pflug and Werner R\"omisch.
\newblock {\em Modeling, Measuring and Managing Risk}.
\newblock World Scientific, River Edge, NJ, 2007.

\bibitem[RU00]{RockafellarUryasev2000}
R.~Tyrrell Rockafellar and Stanislav Uryasev.
\newblock Optimization of {C}onditional {V}alue-at-{R}isk.
\newblock {\em Journal of Risk}, 2:21--41, 2000.

\bibitem[RU02]{Rockafellar}
R.~Tyrrell Rockafellar and Stanislav Uryasev.
\newblock Conditional value-at-risk for general loss distributions.
\newblock {\em Journal of Banking and Finance}, 26:1443--1471, 2002.

\bibitem[SDR09]{RuszczynskiShapiro2009}
Alexander Shapiro, Darinka Dentcheva, and Andrzej Ruszczy\'nski.
\newblock {\em Lectures on {S}tochastic {P}rogramming}.
\newblock MQS-SIAM Series on Optimization 9, 2009.

\bibitem[Sha12]{Shapiro2011}
Alexander Shapiro.
\newblock On {K}usuoka representations of law invariant risk measures.
\newblock {\em Mathematics of Operations Research}, November 2012.

\bibitem[SS07]{shanked}
Moshe Shaked and J.~George Shanthikumar.
\newblock {\em Stochastic Order}.
\newblock Springer Series in Statistics. Springer, 2007.

\bibitem[vdV98]{vdVaart}
Aad~W. van~der Vaart.
\newblock {\em Asymptotic Statistics}.
\newblock Cambridge University Press, 1998.

\bibitem[Vil03]{Villani2003}
Cédric Villani.
\newblock {\em Topics in {O}ptimal {T}ransportation}, volume~58 of {\em
  Graduate Studies in Mathematics}.
\newblock American Mathematical Society, Providence, RI, 2003.

\bibitem[VX11]{Valdez2011}
Emiliano~A. Valdez and Yugu Xiao.
\newblock On the distortion of a copula and its margins.
\newblock {\em Scandinavian Actuarial Journal}, 4:292--317, 2011.

\bibitem[Wan00]{Wang2000}
Shaun~S. Wang.
\newblock A class of distortion operatiors for financial and insurance risk.
\newblock {\em The Journal of Risk and Insurance}, 67(1):15--36, 2000.

\bibitem[WD98]{Wang1998a}
Shaun Wang and Jan Dhaene.
\newblock Comonotonicity, correlation order and premium principles.
\newblock {\em Insurance: Mathematics and Economics}, 22:235--242, 1998.

\bibitem[WY98]{Wang1998}
Shaun~S. Wang and Virginia~R. Young.
\newblock Risk-adjusted credibility premiums using distorted probabilities.
\newblock {\em Scandinavian Actuarial Journal}, 2:143--165, 1998.

\bibitem[WYP97]{Wang1997}
Shaun~S. Wang, Virginia~R. Young, and Harry~H. Panjer.
\newblock Axiomatic characterization of insurance prices.
\newblock {\em Insurance: Mathematics and Economics}, 21:173--183, 1997.

\bibitem[You06]{Young}
Virginia~R. Young.
\newblock {\em Premium Principles}.
\newblock John Wiley \& Sons, Ltd, 2006.

\end{thebibliography}

\appendix

\section*{Appendix}

For reference and the sake of completeness we list the following elementary
result for affine linear transformations of the convex conjugate function.
\begin{lem}
\label{lem:Transform}The convex conjugate of the function $g\left(x\right):=\alpha+\beta x+\gamma\cdot f\left(\lambda x+c\right)$
for $\gamma>0$ and $\lambda\neq0$ is 
\[
g^{*}\left(y\right)=-\alpha-c\,\frac{y-\beta}{\lambda}+\gamma\cdot f^{*}\left(\frac{y-\beta}{\lambda\gamma}\right).
\]
\end{lem}
\begin{proof}
Just observe that 
\begin{align}
g^{*}\left(y\right)= & \sup_{x}\, yx-g(x)\nonumber \\
= & \sup_{x}\, yx-\alpha-\beta x-\gamma\cdot f\left(\lambda x+c\right)\nonumber \\
= & \sup_{x}\, y\frac{x-c}{\lambda}-\alpha-\beta\frac{x-c}{\lambda}-\gamma\cdot f\left(x\right)\label{eq:drei}\\
= & -\alpha-c\frac{y-\beta}{\lambda}+\sup_{x}\, x\frac{y-\beta}{\lambda}-\gamma\cdot f\left(x\right)\nonumber \\
= & -\alpha-c\frac{y-\beta}{\lambda}+\gamma\cdot\sup_{x}\, x\frac{y-\beta}{\lambda\gamma}-f\left(x\right)\nonumber \\
= & -\alpha-c\frac{y-\beta}{\lambda}+\gamma\cdot f^{*}\left(\frac{y-\beta}{\lambda\gamma}\right),\nonumber 
\end{align}
where we have replaced $x$ by $\frac{x-c}{\lambda}$ in \eqref{eq:drei}.\end{proof}

\end{document}